\documentclass[11pt,draftcls,journal,onecolumn]{IEEEtran}
\usepackage{amsmath}
\usepackage{amssymb}
\usepackage{verbatim}
\usepackage{epsfig}
\usepackage{subfigure}
\usepackage{graphicx}
\usepackage{cases}
\usepackage{amsfonts}
\usepackage{cite}

\newtheorem{definition}{Definition}

\newtheorem{theorem}{Theorem}

\newtheorem{remark}{\indent \bf Remark}

\newtheorem{proposition}{Proposition}

\def\snr    {\mbox{\scriptsize\sf SNR}}

\begin{document}
\title{Coding versus ARQ in Fading Channels: How reliable should the PHY be?}

\author{Peng~Wu~and~Nihar~Jindal\\
University of Minnesota, Minneapolis, MN 55455\\
Email: \{pengwu,nihar\}@umn.edu}
\maketitle
\vspace{-15mm}

\begin{abstract}\label{sec-abs}
This paper studies the tradeoff between channel coding and ARQ (automatic repeat request) in Rayleigh block-fading
channels. A heavily coded system corresponds to a low transmission rate with few ARQ re-transmissions, whereas lighter
coding corresponds to a higher transmitted rate but more re-transmissions. The optimum error probability, where optimum refers
to the maximization of the average successful throughput, is derived and is shown to be a decreasing function
of the average signal-to-noise ratio and of the channel diversity order. A general conclusion of the work is that the optimum error
probability is quite large (e.g., $10\%$ or larger) for reasonable channel parameters, and that operating at a very small error
probability can lead to a significantly reduced throughput.  This conclusion holds even when a number of practical ARQ considerations, such as delay constraints and acknowledgement feedback errors, are taken into account.
\end{abstract}

\section{Introduction}\label{sec-intro}
In contemporary wireless communication systems, ARQ (automatic repeat request) is generally used above the physical layer (PHY) to
compensate for packet errors: incorrectly decoded packets are detected by the receiver, and a negative acknowledgement is sent back to the transmitter to request a re-transmission.
In such an architecture there is a natural tradeoff between the transmitted rate and ARQ re-transmissions.  A high transmitted rate corresponds
to many packet errors and thus many ARQ re-transmissions, but each successfully received packet contains many information bits.
On the other hand, a low transmitted rate corresponds to few ARQ re-transmissions, but few information bits are
contained per packet.  Thus, a fundamental design challenge is determining the transmitted rate that maximizes the rate at which
bits are successfully delivered.  Since the packet error probability is an increasing function of the transmitted rate,
this is equivalent to determining the optimal packet error probability, i.e., the optimal PHY reliability level.

We consider a wireless channel where the transmitter chooses the rate based only on the fading statistics because
knowledge of the instantaneous channel conditions is not available (e.g., high velocity mobiles in cellular systems).
The transmitted rate-ARQ tradeoff is interesting in this setting because the packet error probability depends on the transmitted rate
in a non-trivial fashion; on the other hand, this tradeoff is somewhat trivial when instantaneous channel state information at the transmitter (CSIT) is available (see Remark \ref{remk1}).

We begin by analyzing an idealized system, for which we find that making the PHY too reliable can lead to a significant
penalty in terms of the achieved goodput (long-term average successful \emph{throughput}), and that the optimal packet error
probability is decreasing in the average SNR and in the fading selectivity experienced by each transmitted codeword.
We also see that for a large level of system parameters, choosing an error probability of $10\%$ leads to near-optimal
performance.  We then consider a number of important practical considerations, such as a limit on the number of ARQ re-transmissions
and unreliable acknowledgement feedback.  Even after taking these issues into account, we find that a relatively
unreliable PHY is still preferred.  Because of fading, the PHY can be made reliable only if the transmitted rate is significantly
reduced.  However, this reduction in rate is not made up for by the corresponding reduction in ARQ re-transmissions.

%
%
%

\subsection{Prior Work}

There has been some recent work on the joint optimization of packet-level erasure-correction codes (e.g., fountain codes)
and PHY-layer error correction \cite{luby2007rmd,berger2008optimizing,Courtade,ChenSub}.  The fundamental metric with erasure codes is the product of the transmitted rate and the packet success probability, which is the same as in the idealized ARQ setting studied in Section \ref{sec-ideal}.
Even in that idealized setting, our work differs in a number of ways.  References \cite{luby2007rmd,Courtade,ChenSub}  study
 multicast (i.e., multiple receivers) while \cite{berger2008optimizing} considers unicast assuming no diversity per transmission, whereas
 our focus is on the unicast setting with diversity per transmission.  Furthermore, our analysis provides a general explanation of
 how the PHY reliability should depend on both the diversity and the average SNR.  In addition, we consider a number of practical issues
 specific to ARQ, such as acknowledgement errors (Section \ref{sec-nonideal}), as well as hybrid-ARQ (Section \ref{sec-harq}).

\section{System Model}\label{sec-sys}
We consider a Rayleigh block-fading channel where the channel remains constant within each block but changes independently from one block to another.
The $t$-th ($t=1,2,\cdots$) received channel symbol in the $i$-th ($i=1,2, \cdots$) fading block $y_{t,i}$ is given by
\begin{eqnarray}
y_{t,i} = \sqrt{\snr}~h_{i}x_{t,i} + z_{t,i},
\end{eqnarray}
where $h_i\sim\mathcal{CN}(0,1)$ represents the channel gain and is i.i.d. across fading blocks, $x_{t,i}\sim\mathcal{CN}(0,1)$ denotes the Gaussian input symbol constrained to have unit average power, and $z_{t,i}\sim\mathcal{CN}(0,1)$ models the additive Gaussian noise assumed to be i.i.d. across channel uses and fading blocks. Although we focus on single antenna systems and Rayleigh fading channel, our model can be easily extended to multiple-input and multiple-output (MIMO) systems and other fading distributions as commented upon in Remark \ref{remk3}.

Each transmission (i.e., codeword) is assumed to span $L$ fading blocks, and thus $L$ represents the time/frequency selectivity
experienced by each codeword.  In analyzing ARQ systems, the packet error probability is the key quantity. If a strong channel code (with suitably
long blocklength) is used, it is well known that the packet error probability is accurately approximated by the mutual
information outage probability \cite{CaTaBi,FabregasCaire,PrasadVaranasi_IT06,Malkamaki}.  Under this assumption (which is examined in
Section \ref{sec-finite_blk}), the packet error
probability for transmission at rate $R$ bits/symbol is given by \cite[eq (5.83)]{TseVis}:
\begin{eqnarray}\label{outage_nHARQ}
\varepsilon(\snr,L,R) = \mathbb{P} \left[\frac{1}{L}\sum_{i=1}^L\log_2(1+\snr|h_i|^2)\leq R \right].
\end{eqnarray}
Here we explicitly denote the dependence of the error probability on the average signal-to-noise ratio $\snr$, the selectivity order $L$, and the
transmitted rate $R$.  We are generally interested in the relationship between $R$ and $\varepsilon$ for particular (fixed) values of $\snr$ and
$L$.  When $\snr$ and $L$ are constant, $R$ can be inversely computed given some $\varepsilon$; thus,
throughout the paper we replace $R$ with $R_{\varepsilon}$ wherever the relationship between $R$ and $\varepsilon$ needs to be explicitly pointed out.

The focus of the paper is on simple ARQ, in which packets received in error are re-transmitted and decoding is performed only on the basis
of the most recent transmission.\footnote{\emph{Hybrid}-ARQ, which is a more sophisticated and powerful form of ARQ, is considered in Section \ref{sec-harq}.}  More specifically, whenever the receiver detects that a codeword has been decoded incorrectly, a NACK is fed back to the transmitter.
On the other hand, if the receiver detects correct decoding an ACK is fed back.  Upon reception of an ACK, the transmitter moves on to the next
packet, whereas reception of a NACK triggers re-transmission of the previous packet.  ARQ transforms the system into a variable-rate scheme,
and the relevant performance metric is the rate at which packets are \textit{successfully} received.  This quantity is generally
referred to as the long-term average \emph{goodput}, and is clearly defined in each of the relevant sections. And consistent with the assumption of no CSIT (and fast fading), we assume fading is independent across re-transmissions.

\section{OPTIMAL PHY Reliability in the Ideal Setting}\label{sec-ideal}

In this section we investigate the optimal PHY reliability level under a number of idealized assumptions.  Although not entirely realistic,
this idealized model yields important design insights. In particular, we make the following key assumptions:
\begin{itemize}
\item Channel codes that operate at the mutual information limit (i.e., packet error probability is equal to the mutual information outage probability).
\item Perfect error detection at the receiver.
\item Unlimited number of ARQ re-transmissions.
\item Perfect ACK/NACK feedback.
\end{itemize}
In Section \ref{sec-nonideal} we relax these assumptions, and find that the insights from this idealized setting generally also apply to real systems.

In order to characterize the long-term goodput in this idealized setting.  In order to do so, we must quantify the number of
transmission attempts/ARQ rounds needed for successful transmission of each packet.  If we use $X_i$ to denote the number of ARQ rounds for the \emph{i}-th packet, then a total of $\sum_{i=1}^J X_i$ ARQ rounds are used for transmitting $J$ packets; note that the $X_i$'s are i.i.d. due to the independence of fading and noise across ARQ rounds. Each codeword is assumed to span $n$ channel symbols and to contain $b$ information bits, corresponding
to a transmitted rate of $R=b/n$ bits/symbols.  The average rate at which bits are successfully delivered is the ratio of the bits delivered to the total number of channel symbols required.  The goodput $\eta$ is the long-term average at which bits are successfully delivered, and
by taking $J\rightarrow\infty$ we get \cite{CaireTuninetti}:
\begin{eqnarray}\label{gput_nHARQ}
\eta = \lim_{J\rightarrow\infty}\frac{Jb}{n\sum_{i=1}^JX_i}=\lim_{J\rightarrow\infty}\frac{\frac{b}{n}}{\frac{1}{J}\sum_{i=1}^JX_i}=\frac{R}{\mathbb{E}[X]},
\end{eqnarray}
where $X$ is the random variable describing the ARQ rounds required for successful delivery of a packet.

Because each ARQ round is successful with probability $1-\varepsilon$, with $\varepsilon$ defined in (\ref{outage_nHARQ}),
and rounds are independent, $X$ is geometric
with parameter $1-\varepsilon$ and thus $\mathbb{E}[X]=1/(1-\varepsilon)$. Based upon (\ref{gput_nHARQ}), we have
\begin{eqnarray}\label{gput_nHARQ2}
\eta\triangleq R_{\varepsilon}(1-\varepsilon),
\end{eqnarray}
where the transmitted rate is denoted as $R_{\varepsilon}$ to emphasize its dependence on $\varepsilon$.

Based on this expression, we can immediately see the tradeoff between the transmitted rate, i.e. the number of bits per packet,
and the number of ARQ re-transmissions per packet: a large $R_{\varepsilon}$ means many bits are contained in each
packet but that many re-transmissions are required, whereas a small $R_{\varepsilon}$ corresponds to fewer bits per packet
and fewer re-transmissions. Our objective is to find the optimal (i.e., goodput maximizing) operating point on this tradeoff curve
for any given parameters $\snr$ and $L$.

Because $R_{\varepsilon}$ is a function of $\varepsilon$ (for $\snr$ and $L$ fixed), this
one-dimensional optimization can be phrased in terms of $R_{\varepsilon}$ or $\varepsilon$.  We
find it most insightful to consider $\varepsilon$, which leads to the following definition:

\begin{definition} \label{def1}
The optimal packet error probability, where optimal refers to goodput maximization with goodput defined in (\ref{gput_nHARQ}),
for average signal-to-noise ratio $\snr$ and per-codeword selectivity order $L$ is:
\begin{eqnarray}\label{opt_gput_nHARQ}
\varepsilon^{\star}(\snr,L)\triangleq\arg \max_{\varepsilon} ~ R_{\varepsilon}(1-\varepsilon).
\end{eqnarray}
\end{definition}

By finding $\varepsilon^{\star}(\snr,L)$, we thus determine the optimal PHY reliability level and how this optimum depends on
channel parameters $\snr$ and $L$, which are generally static over the timescale of interest.\footnote{Note that in this definition
we assume all possible code rates are possible; nonetheless, this formulation provides valuable insight for systems in which
the transmitter must choose from a finite set of code rates.}

For $L=1$, a simple calculation shows \footnote{The expression for $L=1$ is also derived in \cite{berger2008optimizing}. However, authors in \cite{berger2008optimizing} only consider $L=1$ case rather than $L>1$ scenarios, which are further investigated in our work.}
\begin{eqnarray}
\varepsilon^{\star}(\snr,1) = 1- e^{\left(1-\snr\right)/\left(\snr\cdot W(\snr)\right)},
\end{eqnarray}
where $W(\cdot)$ is the Lambert W function \cite{corless1996lambertw}. Unfortunately, for $L>1$ it does not seem feasible to find an exact analytical solution because a closed-form expression for the outage probability exists only for $L=1$. However, the optimization in (\ref{opt_gput_nHARQ}) can be easily solved numerically (for arbitrary $L$).  In addition, an accurate approximation to $\varepsilon^{\star}(\snr,L)$ can be solved analytically, as we detail
in the next subsection.

In order to provide a general understanding of $\varepsilon^{\star}$, Fig. \ref{fig:gput_eps} contains a plot of goodput $\eta$ (numerically computed)
versus outage probability $\varepsilon$ for $L=2$ and $L=5$ at $\snr=0$ and $10$ dB. For each curve, the goodput-maximizing value of $\varepsilon$ is circled. From this figure, we make the following observations:
\begin{itemize}
\item Making the physical layer too reliable or too unreliable yields poor goodput.
\item The optimal outage probability decreases with $\snr$ and $L$.
\end{itemize}
These turn out to be the key behaviors of the coding-ARQ tradeoff, and the remainder of this section is devoted to analytically explain
these behaviors through a Gaussian approximation.

\begin{remark}\label{remk1}
Throughput the paper we consider the setting \emph{without} channel state information at the transmitter (CSIT). If there is CSIT,
which generally is the case when the fading is slow relative to the delay in the channel feedback loop, the optimization problem in
\emph{Definition \ref{def1}} turns out to be trivial. When CSIT is available, the channel is essentially AWGN with an instantaneous SNR that is
determined by the fading realization but is known to the TX.
If a capacity-achieving code with infinite codeword block-length is used in the AWGN channel,
the relationship between error and rate is a step-function:
\begin{subnumcases}{\varepsilon=}
0, & if $R<\log_2\left(1+\snr|h|^2\right)$\\
1, & if $R\geq \log_2\left(1+\snr|h|^2\right)$.
\end{subnumcases}
Thus, it is optimal to choose a rate very slightly below the instantaneous capacity $\log_2\left(1+\snr|h|^2\right)$. For
realistic codes with finite blocklength, the $\varepsilon$-$R$ curve is not a step function but nonetheless is very steep.
For example, for turbo codes the waterfall characteristic of error vs. SNR curves (for fixed rate) translates to a
step-function-like error vs. rate curve for fixed SNR. Therefore, the transmitted rate should be chosen close to the bottom
of the step function.
\end{remark}


\subsection{Gaussian Approximation}\label{sec_gauss}

The primary difficulty in finding $\varepsilon^{\star}(\snr,L)$
stems from the fact that the outage probability in (\ref{outage_nHARQ}) can
only be expressed as an $L$-dimensional integral, except for the special
case $L=1$.  To circumvent this problem, we utilize a Gaussian approximation
to the outage probability used in prior work \cite{Smith_Shafi, Barriac, wu2008}.
The random variable $\frac{1}{L}\sum_{i=1}^L\log_2\left(1+\snr |h_{i}|^2\right)$ is approximated
by a $\mathcal{N}\left(\mu(\snr),\sigma^2(\snr)/L\right)$ random variable, where
$\mu(\snr)$ and $\sigma^2(\snr)$ are the mean and the variance of
$\log_2\left(1+\snr |h|^2\right)$, respectively:
\begin{eqnarray}
\mu(\snr) &=& \mathbb{E}_{|h|}\left[\log_2(1+\snr|h|^2)\right],\\
\sigma^2(\snr)&=&\mathbb{E}_{|h|}\left[\log_2(1+\snr|h|^2)\right]^2-\mu^2(\snr).
\end{eqnarray}
Closed forms for these quantities can be found in \cite{Alouini, McKay}.
Based on this approximation we have
\begin{eqnarray}
\varepsilon&\approx&
Q\left(\frac{\sqrt{L}}{\sigma(\snr)}(\mu(\snr)-R_{\varepsilon})\right),
\end{eqnarray}
where $Q(\cdot)$ is the tail probability of a standard normal.
Solving this equation for $R_{\varepsilon}$ and plugging into (\ref{gput_nHARQ2})
yields the following approximation for the goodput, which we denote as $\eta_g$:
\begin{eqnarray}\label{thput_app}
\eta_g = \left(\mu(\snr)-Q^{-1}(\varepsilon)\frac{\sigma(\snr)}{\sqrt{L}}\right)(1-\varepsilon),
\end{eqnarray}
where $Q^{-1}(\varepsilon)$ is the inverse of the $Q$ function.

\subsection{Optimization of Goodput Approximation}\label{sec_opt}
The optimization of $\eta_g$ turns out to be more tractable.  We first rewrite $\eta_g$ as
\begin{eqnarray} \label{eta_g_2}
\eta_g = \mu(\snr)\left(1-\kappa \cdot Q^{-1}(\varepsilon)\right)(1-\varepsilon),
\end{eqnarray}
where the constant $\kappa \in(0,1)$ is the $\mu$-normalized standard deviation
of the  received mutual information:
\begin{equation}
\kappa \triangleq \frac{\sigma(\snr)}{\mu(\snr)\sqrt{L}}.
\end{equation}
We can observe that $\kappa$ decreases in $\snr$ and $L$. We now define $\varepsilon_g^{\star}$ as the $\eta_g$-maximizing outage probability:
\begin{eqnarray}\label{opt_no_control_app}
\varepsilon_g^{\star}(\snr,L)\triangleq\arg\max_{\varepsilon} ~ \left(1-\kappa \cdot Q^{-1}(\varepsilon)\right)(1-\varepsilon),
\end{eqnarray}
where we have pulled out the constant $\mu(\snr)$ from (\ref{eta_g_2}) because
it does not affect the maximization.
\begin{proposition}\label{pro1}
The PHY reliability level that maximizes the Gaussian approximated goodput is the unique solution to the following fixed point equation:
\begin{eqnarray}\label{first_der}
\left(Q^{-1}(\varepsilon_g^{\star})-(1-\varepsilon_g^{\star})\cdot\left(Q^{-1}(\varepsilon)\right)'\mid_{\varepsilon=\varepsilon_g^{\star}}\right)^{-1}=\kappa.
\end{eqnarray}
Furthermore, $\varepsilon_g^{\star}$ is increasing in $\kappa$.
\end{proposition}
\begin{proof}
See Appendix \ref{pf_a}.
\end{proof}
We immediately see that $\varepsilon_g^{\star}$ depends on the channel parameters only through $\kappa$. Furthermore,
because $\kappa$ is decreasing in $\snr$ and $L$,  we see that $\varepsilon_g^{\star}$ decreases in $L$ (i.e., the channel selectivity) and
$\snr$.  Straightforward analysis shows that  $\varepsilon_g^{\star}$ tends to zero as $L$ increases approximately as $1/\sqrt{L \log L}$,
while  $\varepsilon_g^{\star}$ tends to zero with $\snr$ approximately as $1/\sqrt{\log \snr}$.


In Fig. \ref{fig:opt_eps_SNR}, the exact optimal $\varepsilon^{\star}$ and
the approximate-optimal $\varepsilon_g^{\star}$ are plotted vs. $\snr$ (dB)
for $L=2,5,$ and $10$. The Gaussian approximation is seen to be reasonably accurate, and
most importantly, correctly captures behavior with respect to $L$ and $\snr$.

In order to gain an intuitive understanding of the optimization,
in Fig. \ref{fig:phy_arg} the success probability $1-\varepsilon$ (left) and the
goodput $\eta=R_{\varepsilon}(1-\varepsilon)$ (right) are plotted versus the transmitted rate $R$
for $\snr=10$ dB. For each $L$ the goodput-maximizing operating point is circled.
First consider the curves for $L=5$.
For $R$ up to approximately $1.5$ bits/symbol the success probability is nearly one, i.e., $\varepsilon \approx 0$.
As a result, the goodput $\eta$ is approximately equal to $R$
for $R$ up to $1.5$.  When $R$ is increased beyond $1.5$ the success probability begins
to decrease non-negligibly but the goodput nonetheless increases with $R$ because the increased
transmission rate makes up for the loss in success probability (i.e., for the
ARQ re-transmissions).  However, the goodput peaks at $R = 2.3$ because beyond this point the
increase in transmission rate no longer makes up for the increased re-transmissions; visually,
the optimum rate (for each value of $L$) corresponds to a point beyond which the
success probability begins to drop off sharply with the transmitted rate.

To understand the effect of the selectivity order $L$, notice that increasing $L$ leads to a steepening of
the success probability-rate curve. This has the effect of moving the goodput curve closer to the transmitted
rate, which leads to a larger optimum rate and a larger optimum success probability ($1-\varepsilon^{\star}$).
To understand why $\varepsilon^{\star}$ decreases with $\snr$, based upon the
rewritten version of $\eta_g$ in (\ref{eta_g_2}) we see that the governing
relationship is between the success probability $1-\varepsilon$ and the
normalized, rather than absolute, transmission rate $R / \mu(\snr)$.
Therefore, increasing $\snr$ steepens the success probability-normalized
rate curve (similar to the effect of increasing $L$) and thus leads to a smaller
value of $\varepsilon^{\star}$.

Is is important to notice that the optimum error probabilities in
Fig. \ref{fig:opt_eps_SNR} are quite large, even for large selectivity and
at high SNR levels.  This follows from the earlier explanation that
decreasing the error probability (and thus the rate) beyond a certain point is
inefficient because the decrease in ARQ re-transmissions does not make up
for the loss in transmission rate.

To underscore the importance of not operating the PHY too reliably, in Fig.
\ref{fig:goodput_snr} goodput is plotted versus $\snr$ (dB) for $L=2$ and $10$
for the optimum error probability $\eta(\varepsilon^{\star})$ as well
as for $\varepsilon=0.1$, $0.01$, and $0.001$. Choosing $\varepsilon=0.1$ leads
to near-optimal performance for both selectivity values.  On the other hand,
there is a significant penalty if $\varepsilon=0.01$ or $0.001$ when $L=2$;
this penalty is reduced in the highly selective channel ($L=10$) but is
still non-negligible. Indeed, the most important insight from this analysis is that making the PHY
too reliable can lead to a significant performance penalty; for example, choosing
$\varepsilon=0.001$ leads to a power penalty of approximately $10$ dB for $L=2$ and
$2$ dB for $L=10$.

\begin{remark}\label{remk3}
\emph{Proposition \ref{pro1}} shows $\varepsilon_g^{\star}$ is only determined by
$\kappa$, which is completely determined by the statistics of the received mutual information per packet. This implies our results can be easily extended to different fading distributions and to MIMO by appropriately modifying $\mu(\snr)$ and $\sigma(\snr)$.
\end{remark}

\section{OPTIMAL PHY Reliability in the Non-ideal Setting}\label{sec-nonideal}
While the previous section illustrated the need to operate the PHY at a relatively unreliable level under a number of idealized assumptions, a legitimate question is whether that conclusion still holds when the idealizations of Section \ref{sec-ideal} are removed. Thereby motivated, in this section we begin to carefully study the following scenarios one by one:
\begin{itemize}
\item Finite codeword block-length.
\item Imperfect error detection.
\item Limited number of ARQ rounds per packet.
\item Imperfect ACK/NACK feedback.
\end{itemize}
As we shall see, our basic conclusion is upheld even under more realistic assumptions.

\subsection{Finite Codeword Block-length}\label{sec-finite_blk}
Although in the previous section we assumed operation at the mutual information of infinite blocklength codes, real systems must use finite blocklength codes. In order to determine the effect of finite blocklength upon the optimal PHY reliability, we study the mutual
information outage probability in terms of the \textit{information spectrum}, which captures the block error probability for finite blocklength
codes.  In \cite{buckingham2008information}, it was shown that actual codes perform quite close to the information spectrum-based outage probability.


By extending the results of \cite{laneman2006distribution, buckingham2008information}, the outage probability with blocklength $n$ (symbols) is
\begin{eqnarray}\label{lemma1}
\varepsilon(n,\snr,L,R)=\mathbb{P}\left[\frac{1}{L}\sum_{i=1}^L\log\left(1+|h_i|^2\snr\right)+\frac{1}{n}\sum_{i=1}^L\left(\sqrt{\frac{|h_i|^2\snr}{1+|h_i|^2\snr}}\cdot\sum_{j=1}^{n/L}\omega_{ij}\right)
\leq R \right],
\end{eqnarray}
where $R$ is the transmitted rate in nats/symbol, and $\omega_{i,j}$'s are i.i.d. Laplace random variables \cite{laneman2006distribution}, each with
 zero mean and variance two.
The first term in the sum is the standard infinite blocklength mutual information expression, whereas the second term is due
to the finite blocklength, and in particular captures the effect of atypical noise realizations.  This second term goes
to zero as $n\rightarrow\infty$ (i.e., atypical noise does not occur in the infinite blocklength limit), but cannot be ignored
for finite $n$.

The sum of i.i.d. Laplace random variables has a Bessel-K distribution, which is difficult to compute for large $n$ but can be
very accurately approximated by a Gaussian as verified in \cite{buckingham2008information}.  Thus, the mutual information
conditioned on the $L$ channel realizations is approximated by a Gaussian random variable:

\begin{eqnarray}\label{gauss_finite}
\mathcal{N}\left(\frac{1}{L}\sum_{i=1}^L\log\left(1+|h_i|^2\snr\right),\frac{1}{L}\sum_{i=1}^L\frac{2|h_i|^2\snr}{n\left(1+|h_i|^2\snr\right)}\right)
\end{eqnarray}
(This is different from Section \ref{sec_gauss}, where the Gaussian approximation is made with respect to the fading realizations).
Therefore, we can approximate the outage probability with finite block-length $n$ by averaging the cumulative distribution function (CDF) of (\ref{gauss_finite}) over different channel realizations:
\begin{eqnarray}
\varepsilon(n,\snr,L,R)\approx\mathbb{E}_{|h_1|, \ldots, |h_L|} Q\left(\frac{\frac{1}{L}\sum_{i=1}^L\log\left(1+|h_i|^2\snr\right)-R}{\sqrt{\frac{1}{L}\sum_{i=1}^L\frac{2|h_i|^2\snr}{n\left(1+|h_i|^2\snr\right)}}}\right).
\end{eqnarray}

In Fig. \ref{fig:ps_R_fin}, we compare finite and infinite blocklength codes by plotting success probability $1-\varepsilon$ vs. $R_{\varepsilon}$ (bits/symbol) for $L=10$ at $\snr=0$ and $10$ dB. It is clearly seen that the steepness of the success-rate curve is reduced by the finite blocklength; this is
a consequence of atypical noise realizations.

We can now consider goodput maximization for a given blocklength $n$:
\begin{eqnarray}
\varepsilon^{\star} (\snr,L,n) \triangleq ~\arg \max_{\varepsilon} R_{\varepsilon}(1-\varepsilon),
\end{eqnarray}
where both $R_{\varepsilon}$ and $\varepsilon$ are computed (numerically) in the finite codeword block-length regime.

In Fig. \ref{fig:opt_eps_snr_fin}, the optimal $\varepsilon$ vs. SNR (dB) is plotted for both finite block-length coding and infinite block-length coding.
We see that the optimal error probability becomes larger, as expected by success-rate curves with reduced steepness in Fig. \ref{fig:ps_R_fin}.
At high SNR, the finite block-length coding curve almost overlaps the infinite block-length coding curve because the unusual noise term in the mutual information expression is negligible for large values of SNR.
As expected, the optimal reliability level with finite blocklength codes does not differ significantly from the idealized case.

\subsection{Non-ideal Error Detection}\label{sec-err_dec}
A critical component of ARQ is error detection, which is generally performed using a cyclic redundancy check (CRC).  The standard usage of CRC corresponds to appending $k$ parity check
bits to $b-k$ information bits, yielding a total of $b$ bits that are then encoded (by the channel encoder) into $n$ channel symbols.
At the receiver, the channel decoder (which is generally agnostic to CRC) takes the $n$ channel symbols as inputs and produces
an estimate of the $b$ bits, which are in turn passed to the CRC decoder for error detection.
A basic analysis in \cite{Gamal} shows that if the channel decoder is in error (i.e., the $b$ bits input to the channel encoder do not match the
$b$ decoded bits), the probability of an undetected error (i.e., the CRC decoder signals correct
even though an error has occurred) is roughly $2^{-k}$.  Therefore, the overall probability of an undetected error is
well approximated by $\varepsilon\cdot 2^{-k}$.

Undetected errors can lead to significant problems, whose severity depends upon higher network layers (e.g., whether or not an additional layer of error
detection is performed at a higher layer) and the application.  However, a general perspective is provided by imposing a constraint $p$ on the
undetected error probability, i.e., $\varepsilon \cdot 2^{-k}\leq p$. Based on this constraint, we see that the constraint can be met
by increasing $k$, which comes at the cost of overhead, or by reducing the packet error probability $\varepsilon$, which can significantly reduce goodput (Section \ref{sec-ideal}).
The question most relevant to this paper is the following: does the presence of a stringent constraint on undetected error probability
motivate reducing the PHY packet error probability $\varepsilon$?

The relevant quantity, along with the undetected error probability, is the rate at which information bits are correctly delivered, which is:
\begin{eqnarray}
\eta = \frac{b-k}{n}\cdot(1-\varepsilon) = \left(R_{\varepsilon}-\frac{k}{n}\right)\cdot(1-\varepsilon),
\end{eqnarray}
where $R_{\varepsilon}-\frac{k}{n}$ is the effective transmitted rate after accounting for the parity check overhead.
It is then relevant to maximize this rate subject to the constraint on undetected error:\footnote{For the sake of compactness, the dependence of $\varepsilon^{\star}$ and $k^{\star}$ upon $\snr$, $L$ and $n$ is suppressed henceforth, except
where explicit notation is required.}:
\begin{eqnarray}\label{opt_crc}
\left(\varepsilon^{\star},k^{\star}\right)\triangleq &&\arg\max_{\varepsilon,k}~\left(R_{\varepsilon}-\frac{k}{n}\right)\cdot(1-\varepsilon)\\
&&\text{subject to}~~\varepsilon\cdot 2^{-k} \leq p\nonumber
\end{eqnarray}

Although this optimization problem (nor the version based on the Gaussian approximation) is not analytically tractable, it is easy to see that
the solution corresponds to $k^{\star} = \lceil - \log_2(p/\varepsilon^{\star}) \rceil$, where $\varepsilon^{\star}$ is roughly the optimum
packet error probability assuming perfect error detection (i.e. the solution from Section \ref{sec-ideal}).  In other words, the undetected
error probability constraint should be satisfied by choosing $k$ sufficiently large while leaving the PHY transmitted rate nearly untouched.
To better understand this, note that reducing $k$ by a bit requires reducing $\varepsilon$ by a factor of two.
The corresponding  reduction in CRC overhead is very small (roughly $1/n$), while the reduction in the transmitted rate is much larger.
Thus, if we consider the choices of $\varepsilon$ and $k$ that achieve the constraint with equality, i.e.,
$k = -\log_2(p/\varepsilon)$, goodput decreases as $\varepsilon$ is decreased below the packet error probability which is optimal under the assumption of perfect error detection.  In other words, operating the PHY at a more reliable point is not worth the small reduction
in CRC overhead.


\subsection{End-to-End Delay Constraint}\label{sec-delay}

In certain applications such as Voice-over-IP (VoIP), there is a limit on the number of re-transmissions per packet as well as a constraint on the fraction
of packets that are not successfully delivered within this limit.  If such constraints are imposed, it may not be clear
how aggressively ARQ should be utilized.

Consider a system where any packet that fails on its $d$-th attempt is discarded (i.e., at most $d-1$ re-transmissions are allowed),
but at most a fraction $q$ of packets can be discarded, where $q > 0$ is a reliability constraint.
Under these conditions, the probability a packet is discarded is $\varepsilon^d$, i.e., the probability of $d$ consecutive decoding failures,
while the long-term average rate at which packets are successfully delivered still is $R_{\varepsilon}(1-\varepsilon)$.
To understand why the goodput expression is unaffected by the delay limit, note that the number of successfully delivered packets
is equal to the number of transmissions in which decoding is successful, regardless of which packets are transmitted in each slot.  The delay
constraint only affects which packets are delivered in different slots, and thus does not affect the goodput.\footnote{The goodput
expression can alternatively be derived by computing the average number of ARQ rounds per packet (accounting for the limit $d$), and
then applying the renewal-reward theorem \cite{wolff1989stochastic}.}

Since the discarded packet probability is $\varepsilon^d$, the reliability constraint requires $\varepsilon \leq q^{1/d}$. We can thus consider maximization of goodput $R_{\varepsilon}(1-\varepsilon)$ subject to the constraint $\varepsilon \leq q^{1/d}$. Because the goodput is observed to be concave in $\varepsilon$, only two possibilities exist. If $q^{\frac{1}{d}}$ is larger than the optimal value of $\varepsilon$ for the unconstrained problem, then the optimal value of $\varepsilon$ is unaffected by $q$. In the more interesting and relevant case where $q^{\frac{1}{d}}$ is smaller than the optimal unconstrained $\varepsilon$, then goodput is maximized by choosing $\varepsilon$ equal to the upper bound $q^{\frac{1}{d}}$.

Thus, a strict delay and reliability constraint forces the PHY to be more reliable than in the unconstrained case.  However, amongst all
allowed packet error probabilities, goodput is maximized by choosing the largest.  Thus, although strict constraints do not
allow for very aggressive use of ARQ, nonetheless ARQ should be utilized to the maximum extent possible.


\subsection{Noisy ACK/NACK Feedback}\label{sec-fb_err}

We finally remove the assumption of perfect acknowledgements, and consider the realistic scenario where ACK/NACK feedback is not perfect
and where the acknowledgement overhead is factored in.   The main issue confronted here is the joint optimization of the reliability level of
the forward data channel and of the reverse acknowledgement (feedback/control) channel.  As intuition suggests, reliable communication
is possible only if some combination of the forward and reverse reliability levels is sufficiently large; thus, it is not
clear if operating the PHY at a relatively unreliable level as suggested in earlier sections is appropriate.
The effects of acknowledgement errors can sometimes be reduced through higher-layer mechanisms (e.g., sequence number check), but
in order to shed the most light on the issue of forward/reverse reliability, we focus on an extreme case where acknowledgement
errors are most harmful.  In particular, we consider a setting with delay and reliability constraints as in Section \ref{sec-delay},
and where any NACK to ACK error leads to a packet missing the delay deadline.  We first describe the feedback channel model, and
then analyze performance.

\subsubsection{Feedback Channel Model}
We assume ACK/NACK feedback is performed over a Rayleigh fading channel using a total of $f$ symbols which are distributed on
$L_{\textrm{fb}}$ independently faded subchannels; here $L_{\textrm{fb}}$ is the diversity order of the feedback channel, which need
not be equal to $L$, the forward channel diversity order.  Since the feedback is binary, BPSK is used with the symbol repeated on each sub-channel
 $f/L_{\textrm{fb}}$ times.  For the sake of simplicity, we assume that the feedback channel has the same average SNR as the forward channel, and that the fading on the feedback channel is independent of the fading on the forward channel.

After maximum ratio combining at the receiver, the effective SNR is $(f/L_{\textrm{fb}}) \cdot\snr\cdot\sum_{i=1}^{L_{\textrm{fb}}}|h_i|^2$, where $h_1,\cdots,h_{L_{\textrm{fb}}}$ are the feedback channel fading coefficients.  The resulting probability of error (denoted by $\varepsilon_{\textrm{fb}}$), averaged over the fading realizations,
is \cite{A.Goldsmith}:
\begin{eqnarray}\label{fb_rl}
\varepsilon_{\textrm{fb}} = \left(\frac{1-\nu}{2}\right)^{L_{\textrm{fb}}}\cdot\sum_{j=0}^{L_{\textrm{fb}}-1}{L_{\textrm{fb}}-1+j \choose j}\left(\frac{1+\nu}{2}\right)^{j},
\end{eqnarray}
where $\nu = \sqrt{\frac{(f/L_{\textrm{fb}})\cdot\snr}{1+(f/L_{\textrm{fb}})\cdot\snr}}$.  Clearly, $\varepsilon_{\textrm{fb}}$ is decreasing in $f$ and $\snr$.\footnote{Asymmetric decision regions can be used, in which case $0 \rightarrow 1$ and $1 \rightarrow 0$ errors have unequal probabilities.  However, this
does not significantly affect performance and thus is not considered.}

\subsubsection{Performance Analysis}
In order to analyze performance with non-ideal feedback, we must first specify the rules by which the transmitter and receiver operate.
The transmitter takes precisely the same actions as in Section \ref{sec-delay}: the transmitter immediately moves on to the
next packet whenever an ACK is received, and after receiving $d-1$ consecutive NACK's (for a single packet) it attempts that packet one last
time but then moves on to the next packet regardless of the acknowledgement received for the last attempt.  Of course, the presence of
feedback errors means that the received acknowledgement does not always match the transmitted acknowledgement.  The receiver
also operates in the standard manner, but we do assume that the receiver can always determine whether or not the packet being received
is the same as the packet received in the previous slot, as can be accomplished by a simple correlation; this reasonable assumption
is equivalent to the receiver having knowledge of acknowledgement errors.

In this setup an ACK$\rightarrow$NACK error causes the transmitter to re-transmit the previous packet, instead of moving on to the next packet.  The receiver is able to recognize that an acknowledgement error has occurred (through correlation
of the current and previous received packets), and because it already decoded the packet correctly it does not attempt to decode again.
Instead, it simply transmits an ACK once again.  Thus, each ACK$\rightarrow$NACK error has the relatively benign effect of wasting
one ARQ round.

On the other hand, NACK$\rightarrow$ACK errors have a considerably more deleterious effect because upon reception of an ACK, the transmitter
automatically moves on to the next packet.  Because we are considering a stringent delay constraint, we assume that such a
NACK$\rightarrow$ACK error cannot be recovered from
and thus we consider it as a lost packet that is counted towards the reliability constraint. This is, in some sense,
a worst-case assumption that accentuates the effect of NACK$\rightarrow$ACK errors; some comments related to this point are put forth
at the end of this section.

To more clearly illustrate the model, the complete ARQ process is shown in Fig. \ref{fig:arq_proc} for $d=3$. Each branch
is labeled with the success/failure of the transmission as well as the acknowledgement (including errors).  Circle nodes refer
to states in which the receiver has yet to successfully decode the packet, whereas triangles refer to states in which the
receiver has decoded correctly.  A packet loss occurs if there is a decoding failure followed by a NACK$\rightarrow$ACK error
in the first two rounds, or if decoding fails in all three attempts.  All other outcomes correspond to cases where
the receiver is able to decode the packet in some round, and thus successful delivery of the packet.  In these cases,
however, the number of ARQ rounds depends on the first time at which the receiver can decode and when the ACK is correctly delivered.
(If an ACK is not successfully delivered, it may take up to $d$ rounds before the transmitter moves on to the next packet.)
Notice that after the $d$-th attempt, the transmitter moves on to the next packet regardless of what acknowledgement is received; this
is due to the delay constraint that the transmitter follows.

Based on the figure and the independence of decoding and feedback errors across rounds, the probability that a packet is lost
(i.e., it is not successfully delivered within $d$ rounds) is:
\begin{eqnarray}\label{xi_del}
\xi_d =\varepsilon\cdot \varepsilon_{\textrm{fb}}+\varepsilon^2(1-\varepsilon_{\textrm{fb}})\varepsilon_{\textrm{fb}}+\cdots+\varepsilon^{d-1}(1-\varepsilon_{\textrm{fb}})^{d-2}\varepsilon_{\textrm{fb}}+\varepsilon^d(1-\varepsilon_{\textrm{fb}})^{d-1},
\end{eqnarray}
where the first $d-1$ terms represent decoding failures followed by a NACK$\rightarrow$ACK error
(more specifically, the $l$-th term corresponds to $l-1$ decoding failures and $l-1$ correct NACK transmissions, followed by another
decoding failure and a NACK$\rightarrow$ACK error), and the last term is the probability of $d$ decoding failures and $d-1$
correct NACK transmissions.  If we alternatively compute the success probability, we get the following different expression for $\xi_d$:
\begin{eqnarray}\label{xi_del2}
\xi_d = 1-\sum_{i=1}^d(1-\varepsilon)\cdot\varepsilon^{i-1}\cdot(1-\varepsilon_{\textrm{fb}})^{i-1},
\end{eqnarray}
where the $i$-th summand is the probability that successful forward transmission occurs in the $i$-th ARQ round. Based upon (\ref{xi_del}) and (\ref{xi_del2})
we see that $\xi_d$ is increasing in both $\varepsilon$ and $\varepsilon_{\textrm{fb}}$.
Thus, a desired packet loss probability $\xi_d$ can be achieved by different combinations of the forward channel reliability
and the feedback channel reliability: a less reliable forward channel requires a more reliable feedback channel, and vice versa.

As in Section  \ref{sec-delay} we impose a reliability constraint $\xi_d \leq q$, which by
 (\ref{xi_del}) translates to a joint constraint on $\varepsilon$ and $\varepsilon_{\textrm{fb}}$.
The relatively complicated joint constraint can be accurately approximated by two
much simpler constraints.  Since we must satisfy $\varepsilon\leq q^{\frac{1}{d}}$ even with perfect feedback ($\varepsilon_{\textrm{fb}}=0$), for
 any $\varepsilon_{\textrm{fb}}>0$ we also must satisfy $\varepsilon\leq q^{\frac{1}{d}}$ (this ensures that $d$ consecutive
 decoding failures do not occur too frequently).  Furthermore, by examining (\ref{xi_del}) it is evident that the
 first term is dominant in the packet loss probability expression.  Thus the constraint $\xi_d \leq q$ essentially translates
 to the simplified constraints
\begin{eqnarray} \label{simple-constraints}
\varepsilon \cdot \varepsilon_{\textrm{fb}} \leq q  \textrm{~~~and~~~} \varepsilon\leq q^{\frac{1}{d}}.
\end{eqnarray}
These simplified constraints are very accurate for values of $\varepsilon$ not too close to $q^{\frac{1}{d}}$.
On the other hand, as $\varepsilon$ approaches $q^{\frac{1}{d}}$, $\varepsilon_{\textrm{fb}}$ must go to zero very
rapidly (i.e. much faster than $q/\varepsilon$) in order for $\xi_d \leq q$.

The first constraint in (\ref{simple-constraints}) reveals a general design principle:
the \textit{combination} of the forward and feedback channel must be sufficiently reliable.
This is because $\varepsilon \cdot \varepsilon_{\textrm{fb}}$ is precisely the probability that a packet is lost because
the initial transmission is decoded incorrectly and is followed by a NACK$\rightarrow$ACK error.

Having established the reliability constraint, we now proceed to maximizing goodput while taking acknowledgement errors and ARQ overhead into account.
With respect to the long-term average goodput, by applying the renewal-reward theorem again we obtain:
\begin{eqnarray}\label{gput_db_del}
\eta  &=& \frac{n}{n+f}\cdot \frac{R_{\varepsilon} (1-\xi_d)}{\mathbb{E}[X]}.
\end{eqnarray}
where random variable $X$ is the number of ARQ rounds per packet, and $\mathbb{E}[X]$ is derived in Appendix \ref{pf_c}.
Here, $\frac{n}{n+f}$ is the feedback overhead penalty because each packet spanning $n$ symbols is followed by $f$ symbols to convey the
acknowledgement.

We now maximize goodput with respect to both the forward and feedback channel error probabilities:
\begin{eqnarray}
\left(\varepsilon^{\star},\varepsilon_{\textrm{fb}}^{\star}\right)\triangleq&&\arg\max_{\varepsilon,\varepsilon_{\textrm{fb}}}~~~\frac{n}{n+f}\cdot \frac{R_{\varepsilon}(1-\xi_d)}{\mathbb{E}[X]}\\
&&\text{subject to}~~ \xi_d \leq q\nonumber
\end{eqnarray}
noting that $\varepsilon_{\textrm{fb}}$ is a decreasing function of the number of feedback symbols $f$, according to (\ref{fb_rl}).
This optimization is not analytically tractable, but can be easily solved numerically and can be understood through examination of the
dominant relationships.  The overhead factor $n/(n+f)$ clearly depends only on $\varepsilon_{\textrm{fb}}$ (i.e., $f$).  Although the second term $R_{\varepsilon}(1-\xi_d)/\mathbb{E}[X]$ depends on both $\varepsilon$ and
$\varepsilon_{\textrm{fb}}$, the dependence upon $\varepsilon_{\textrm{fb}}$ is relatively minor as long as $\varepsilon_{\textrm{fb}}$ is reasonably small
(i.e. less than $10\%$).  Thus, it is reasonable to consider the perfect feedback setting, in which case the second term is
$R_{\varepsilon}(1-\varepsilon)$.  Therefore, the challenge is balancing the feedback channel overhead factor $\frac{n}{n+f}$ with the efficiency of the forward channel, approximately $R_{\varepsilon}(1-\varepsilon)$,
while satisfying the constraint in (\ref{simple-constraints}).  If $f$ is chosen small, the feedback errors must be compensated with a
very reliable, and thus inefficient, forward channel; on the other hand, choosing $f$ large incurs a large feedback overhead penalty but
allows for a less reliable, and thus more efficient, forward channel.

In Fig. \ref{fig:eps_theta_del}, the jointly optimal ($\varepsilon^{\star},\varepsilon_{\textrm{fb}}^{\star}$) are plotted
for a conservative set of forward channel parameters  ($L=3$ with $\snr=5$ or $10$ dB, and $n=200$ data symbols per packet), stringent
delay and reliability constraints (up to $d=3$ ARQ rounds and a reliability constraint $q=10^{-6}$), and different diversity orders ($L_{\textrm{fb}}=1,2$ and $5$) for the feedback channel.  Also plotted is the curve specifying the ($\varepsilon,\varepsilon_{\textrm{fb}}$) pairs that achieve the reliability constraint $\xi_d = q$. As discussed earlier, this curve has two distinct regions: for $\varepsilon < 0.008$ it is essentially the straight line $\varepsilon \cdot \varepsilon_{\textrm{fb}} = q$, whereas  $\varepsilon_{\textrm{fb}}$ must go to zero very quickly
as $\varepsilon$ approaches $q^{1/d}=10^{-2}$.

When $L_{\textrm{fb}}=2$, the optimal point corresponds to the transition between these two regions.  Moving to the right
of the optimal corresponds to making the PHY more reliable while making the control channel less reliable (i.e., decreasing
$\varepsilon$ and $f$), but this is suboptimal because the overhead savings do not compensate for the loss incurred by a more reliable PHY.
On the other hand, moving to the left is suboptimal because only a very modest increase in $\varepsilon$ is allowed, and this increase
comes at a large expense in terms of control symbols.  If $L_{\textrm{fb}}=5$, the optimal point is further to the left because the
feedback overhead required to achieve a desired error rate is reduced.
However, the behavior is quite different if there is no diversity on the feedback channel ($L_{\textrm{fb}}=1$).  Without
diversity, the feedback error probability decreases extremely slowly with $f$ (at order $1/f$), and thus a very large $f$
is required to achieve a reasonable feedback error probability.  In this extreme case, it is optimal to sacrifice significant PHY efficiency and choose $\varepsilon$ quite a bit
smaller than $q^{1/d}=10^{-2}$.  Notice that increasing $\snr$ moves the optimal to the left for all values of $L_{\textrm{fb}}$
because a larger SNR improves the feedback channel reliability while not significantly changing the behavior of the forward channel.

This behavior is further explained in Fig. \ref{fig:gput_eps_del}, where goodput $\eta$ (optimized with respect to $\varepsilon_{\textrm{fb}}$)
is plotted versus forward error probability $\varepsilon$ for the parameters of the previous figure, with $\snr=5$ dB and $L_{\textrm{fb}}=1$ and $2$ here.
The figure illustrates the stark contrast with respect to feedback channel diversity: with diversity (even for $L_{\textrm{fb}}=2$), the goodput increases monotonically up to a point quite close to $q^{1/d}$, while without diversity the goodput peaks at a point far below $q^{1/d}$.
This is due to the huge difference in the feedback channel reliability with and without diversity: in order to achieve $\varepsilon_{\textrm{fb}} = 10^{-3}$,
at $\snr=5$ dB without diversity $f=79$ symbols are required, whereas $f=9$ suffices for $L_{\textrm{fb}}=2$.  To more clearly understand why
the optimal point with diversity is so close to $q^{1/d}$, let us contrast two different choices of $\varepsilon$ for $L_{\textrm{fb}}=2$.  At the optimal
$\varepsilon = 8\times10^{-3}$, we require $\varepsilon_{\textrm{fb}}= 6.3\times10^{-5}$ and thus $f=34$.   On the other hand, at the suboptimal
$\varepsilon = 10^{-3}$ we require $\varepsilon_{\textrm{fb}}=10^{-3}$ and thus $f=9$.  Reducing the forward error probability by a factor of
$8$ reduces the feedback overhead from $\frac{34}{234}$ to $\frac{9}{209}$, but reduces the transmitted rate by about $50\%$.

The takeaway message of this analysis is clear: as long as the feedback channel has at least some diversity
(e.g., through frequency or antennas), stringent post-ARQ
reliability constraints should be satisfied by increasing the reliability of the feedback channel instead of increasing the forward
channel reliability.  This is another consequence of the fact that decreasing the forward channel error probability requires
a huge backoff in terms of transmitted rate, which in this case is not compensated by the corresponding decrease in feedback overhead.


\section{Hybrid-ARQ}\label{sec-harq}

While up to now we have considered simple ARQ, contemporary wireless systems often utilize more powerful hybrid-ARQ (HARQ) techniques.
When incremental redundancy (IR) HARQ, which is the most powerful type of HARQ, is implemented, a NACK triggers the transmission
of extra parity check bits instead of re-transmission of the original packet, and the receiver attempts to decode a packet on the basis of all previous
transmissions related to that packet.  This corresponds to accumulation of mutual information across HARQ rounds, and thus essentially
matches the transmitted rate to the instantaneous channel conditions without requiring CSI at the transmitter \cite{CaireTuninetti, wu2008}.
The focus of this section is understanding how the PHY transmitted rate should be chosen when HARQ is used.

Unlike simple ARQ, HARQ requires the receiver to keep information from previous rounds in memory; partly for this reason, HARQ is generally implemented  in a two-layered system (e.g., in 4G cellular networks such as LTE \cite{meyer2006arq} \cite{ekstrom2006technical}) in which the HARQ
process has to restart (triggered by a higher-layer simple ARQ re-transmission) if the number of HARQ rounds reaches a defined maximum.
The precise model we study is described as follows.  As before, each HARQ transmission (i.e., round) experiences a diversity order of $L$.
However, a maximum of $M$ HARQ rounds are allowed per packet.  If  a packet cannot be decoded after $M$ HARQ rounds,
a post-HARQ outage is declared.  This triggers a higher-layer simple ARQ re-transmission, which restarts the HARQ process for that packet.  This
two-layered ARQ process continues (indefinitely) until the packet is successfully received at the receiver.  For the sake of
simplicity, we proceed under the ideal assumptions discussed in Section \ref{sec-ideal}. Note that the case $M=1$ reverts to the simple ARQ model discussed in the rest of the paper.


Given this model, the first-HARQ-round outage probability, denoted $\varepsilon_1$, is exactly the same as the non-HARQ outage probability with the same $\snr$, diversity order $L$, and rate $R$ , i.e.,
\begin{eqnarray}
\varepsilon_1(\snr,L,R)=\mathbb{P}\left[\frac{1}{L}\sum_{i=1}^L\log_2\left(1+\snr|h_i|^2\right)\leq R\right].
\end{eqnarray}
In this expression $R$ is the transmitted rate during the first HARQ round, which we refer to as the HARQ initial rate $R_{\textrm{init}}$ hereafter.
Because IR leads to accumulation of mutual information, the number of HARQ rounds needed to decode a packet is the smallest integer $\mathcal{T}$ ($1\leq \mathcal{T}\leq M$) such that
\begin{eqnarray}\label{defnT}
\sum_{i=1}^{\mathcal{T}}\left(\frac{1}{L}\sum_{j=1}^L\log_2\left(1+\snr|h_{i,j}|^2\right)\right)>R_{\textrm{init}}.
\end{eqnarray}
Therefore, the post-HARQ outage, denoted by $\varepsilon$, is:
\begin{eqnarray}\label{pHARQ}
\varepsilon(\snr,L,M,R_{\textrm{init}}) &=& \mathbb{P} \left[\sum_{i=1}^M\left(\frac{1}{L}\sum_{j=1}^L\log_2\left(1+\snr|h_{i,j}|^2\right)\right)\leq R_{\textrm{init}}\right].
\end{eqnarray}
This is the probability that a packet fails to be decoded after $M$ HARQ rounds, and thus is the probability that the HARQ process
has to be restarted.

Using the renewal-reward theorem as in \cite{CaireTuninetti} yields the following expression for the long-term average goodput with HARQ:
\begin{eqnarray}\label{HARQ_gput}
\eta= \frac{R_{\textrm{init}}(1-\varepsilon)}{\mathbb{E}[\mathcal{T}]},
\end{eqnarray}
where the distribution of $\mathcal{T}$ is determined by (\ref{defnT}).  Our interest is in finding
the initial rate $R_{\textrm{init}}$ that maximizes $\eta$.  This optimization is not analytically tractable, but we can nonetheless
provide some insight.

In Fig. \ref{fig:gput_r_comp}, goodput is plotted versus vs. $R_{\textrm{init}}$ for $L=2$ and a maximum of $M=2$ HARQ rounds, as
well as for a system using only simple ARQ (i.e., $M=1$) with the same $L=2$, at $\snr=5$ and $10$ dB.  We immediately observe
that goodput with HARQ is maximized at a considerably higher rate than for the system without HARQ.
Although we do not have analytical proof, we conjecture that the goodput-maximizing initial rate with HARQ is
always larger than the maximizing rate without HARQ (for equal diversity order per round/transmission).  In fact, with HARQ
the initial rate should be chosen such that the first-round outage $\varepsilon_1$ is quite large, and for larger values of
$M$ the optimizer actually trends towards one.   If $\varepsilon_1$ is small, then HARQ is rarely used which means that
the rate-matching capability provided by HARQ is not exploited.   However, $R_{\textrm{init}}$ should not be chosen so large
such that there is significant probability of post-HARQ outage, because this leads to a simple ARQ re-transmission and thus
forces HARQ to re-start.  The following theorem provides an upper bound on the optimal initial rate:
\begin{theorem}\label{theo1}
For any $\snr, L$, and $M$, the optimal initial rate with HARQ is upper bounded by
$1/M$ times the optimal transmitted rate for a non-HARQ system with diversity order $ML$.
\end{theorem}
\begin{proof}
The HARQ goodput can be rewritten as
\begin{eqnarray} \label{eq-gputharq}
\eta = \frac{R_{\textrm{init}}}{M}\cdot(1-\varepsilon)\cdot \frac{M}{\mathbb{E}[\mathcal{T}]}.
\end{eqnarray}
Based on (\ref{pHARQ}) we see that the post-HARQ outage probability $\varepsilon$ is precisely the same as the outage probability for
a non-HARQ system with diversity order $ML$ and transmitted rate $R_{\textrm{init}}/M$.  Therefore, the term
$(R_{\textrm{init}}/M)(1-\varepsilon)$ in (\ref{eq-gputharq}) is precisely the goodput for a non-HARQ system with diversity order $ML$.
Based on (\ref{defnT}) we can see that the term $M/\mathbb{E}[\mathcal{T}]$ is decreasing in $R_{\textrm{init}}/M$, and thus the value of
$R_{\textrm{init}}/M$ that maximizes (\ref{eq-gputharq}) is smaller than the value that maximizes $(R_{\textrm{init}}/M)(1-\varepsilon)$.
\end{proof}
Notice that $ML$ is the maximum diversity experienced by a packet if HARQ is used, whereas $ML$ is the precise diversity order
experienced by each packet in the reference system (in the theorem) without HARQ.  Combined with our earlier observation, we see that
the initial rate should be chosen large enough such that HARQ is sufficiently utilized, but not so large such that simple ARQ
is overly used.

\section{Conclusion}
In this paper we have conducted a detailed study of the optimum physical layer reliability when simple ARQ is used to
re-transmit incorrectly decoded packets. Our findings show that when a cross-layer perspective is taken, it is optimal to use a rather
unreliable physical layer (e.g., a packet error probability of 10\% for a wide range of channel parameters).
The fundamental reason for this is that making the physical layer very reliable requires a very conservative transmitted rate
in a fading channel (without instantaneous channel knowledge at the transmitter).

 Our findings are quite general, in the
sense that the PHY should not be operated reliably even in scenarios in which intuition might suggest PHY-level reliability is
necessary.  For example, if a smaller packet error mis-detection probability is desired, it is much more efficient to utilize additional
error detection bits (e.g., CRC) as compared to performing additional error correction (i.e., making the PHY more reliable).
A delay constraint imposes an upper bound on the number of ARQ re-transmissions and an upper limit on the PHY error probability,
but an optimized system should operate at exactly this level and no lower.  Finally, when acknowledgement errors are taken into account and high end-to-end
reliability is required, such reliability should be achieved by designing a reliable feedback channel instead of a reliable
data (PHY) channel.

In a broader context, one important message is that traditional diversity metrics, which characterize how quickly the probability of error can be made very small, may no longer be appropriate for wireless systems due to the presence of ARQ. As seen in \cite{NiharAngel} in the context of multi-antenna communication, this change can significantly reduce the attractiveness of transmit diversity techniques that reduce error at the expense of rate.


\appendices
\section{PROOF of Proposition \ref{pro1}}\label{pf_a}
We first prove the strict concavity of $\eta_g$. For any invertible function $f(\cdot)$,  the following holds \cite{apostol1974mathematical}:
\begin{eqnarray}
\left(f^{-1}(a)\right)'=\frac{1}{f'(f^{-1}(a))}.
\end{eqnarray}
By combining this with $Q(x)=\int_x^{\infty}\frac{1}{\sqrt{2\pi}}e^{-\frac{t^2}{2}}dt$, we get
\begin{eqnarray}
\left(Q^{-1}(\varepsilon)\right)'=-\sqrt{2\pi}e^{\frac{(Q^{-1}(\varepsilon))^2}{2}},\label{Qinv_dev}
\end{eqnarray}
which is strictly negative. According to this, the second derivative of $\eta_g(\varepsilon)$ is:
\begin{eqnarray}\label{etag_prov}
\left(\eta_g(\varepsilon)\right)''
&=&\kappa \mu\left(Q^{-1}(\varepsilon)\right)'\left(2+(1-\varepsilon)\sqrt{2\pi}e^{\frac{(Q^{-1}(\varepsilon))^2}{2}}Q^{-1}(\varepsilon)\right).
\end{eqnarray}
Because $\kappa \left(Q^{-1}(\varepsilon)\right)' < 0$, in order to prove $\left(\eta_g(\varepsilon)\right)''<0$
we only need to show that the expression inside the parenthesis in (\ref{etag_prov}) is strictly positive.
If we substitute $\varepsilon=Q(x)$ (here we define $x=Q^{-1}(\varepsilon)$) , then we only need to prove $(Q(x)-1)e^{\frac{x^2}{2}}x<\sqrt{\frac{2}{\pi}}$.
Notice when $x\geq0$, the left hand side is negative (because $Q(x) \leq 1$) and the inequality holds.
When $x<0$, the left hand side becomes $Q(-x)e^{\frac{x^2}{2}}(-x)$. From \cite{Nick}, $Q(-x)<\frac{1}{\sqrt{2\pi}(-x)}e^{-\frac{x^2}{2}}$, so if $x<0$,
\begin{eqnarray}
(Q(x)-1)e^{\frac{x^2}{2}}x<\frac{1}{\sqrt{2\pi}(-x)}e^{-\frac{x^2}{2}}e^{\frac{x^2}{2}}(-x)=\frac{1}{\sqrt{2\pi}}<\sqrt{\frac{2}{\pi}}.
\end{eqnarray}
As a result, the second derivative of $\eta_g(\varepsilon)$ is strictly smaller than zero and thus
$\eta_g$ is strictly concave in $\varepsilon$.
Since $\eta_g$ is strictly concave in $\varepsilon$, we reach the fixed point equation in (\ref{first_der}) by setting the first derivative to zero.
The concavity of $\eta_g$ implies $\left(\eta_g(\varepsilon)\right)'$ is decreasing in $\varepsilon$, and thus
from (\ref{first_der}) we see that $\varepsilon_g^{\star}$ is increasing in $\kappa$.

\section{Expected ARQ Rounds with Acknowledgement Errors}\label{pf_c}
If the ARQ process terminates after $i$ rounds ($1\leq i\leq d-1$), the reasons for that can be:
\begin{itemize}
\item The first $i$ decoding attempts are unsuccessful, the first $i-1$ NACKs are received correctly,
but a NACK$\rightarrow$ACK error happens in the $i$-th round, the probability of which is $\varepsilon^i\cdot(1-\varepsilon_{\textrm{fb}})^{i-1}\cdot\varepsilon_{\textrm{fb}}$.
\item The packet is decoded correctly in the $j$-th round (for $1 \leq j \leq i$), but the ACK is not
correctly received until the $i$-th round.  This corresponds to $j-1$ decoding failures with
correct acknowledgements, followed by a decoding success and $i-j$ acknowledgement errors (ACK$\rightarrow$NACK),
and then a correct acknowledgement: $\sum_{j=1}^i \varepsilon^{j-1}  (1-\varepsilon_{\textrm{fb}})^j
(1 - \varepsilon) \varepsilon_{\textrm{fb}}^{i-j} $.
\end{itemize}
These events are all exclusive, and thus we can sum the above probabilities.  For $X=d$, we notice that the ARQ process
takes the maximum of $d$ rounds if:
\begin{itemize}
\item There are $d$ decoding failures with $d-1$ correct NACKs, the probability of which is $\varepsilon^{d-1}\cdot(1-\varepsilon_{\textrm{fb}})^{d-1}$.
\item The packet is decoded correctly in the $j$-th round (for $1 \leq j \leq d-1$), but the ACK is never
received correctly.  This corresponds to $j-1$ decoding failures with correct NACKs,
followed by a decoding success and $d-j$ acknowledgement errors (ACK$\rightarrow$NACK):
$\sum_{j=1}^{d-1} \varepsilon^{j-1}  (1-\varepsilon_{\textrm{fb}})^{j-1}
(1 - \varepsilon) \varepsilon_{\textrm{fb}}^{d-j} $.
\end{itemize}
These events are again exclusive.  Therefore, the expected number of rounds is:
\begin{eqnarray}\label{es_fb}
\mathbb{E}[X] &=& \sum_{i=1}^{d-1}i\cdot\left(
\varepsilon^i\cdot(1-\varepsilon_{\textrm{fb}})^{i-1}\cdot\varepsilon_{\textrm{fb}} + \sum_{j=1}^i \varepsilon^{j-1}  (1-\varepsilon_{\textrm{fb}})^j
(1 - \varepsilon) \varepsilon_{\textrm{fb}}^{i-j} \right) \nonumber\\
&&+d\cdot\left( \varepsilon^{d-1}\cdot(1-\varepsilon_{\textrm{fb}})^{d-1} + \sum_{j=1}^{d-1} \varepsilon^{j-1}  (1-\varepsilon_{\textrm{fb}})^{j-1}
(1 - \varepsilon) \varepsilon_{\textrm{fb}}^{d-j} \right).
\end{eqnarray}


\begin{thebibliography}{10}
\providecommand{\url}[1]{#1}
\csname url@samestyle\endcsname
\providecommand{\newblock}{\relax}
\providecommand{\bibinfo}[2]{#2}
\providecommand{\BIBentrySTDinterwordspacing}{\spaceskip=0pt\relax}
\providecommand{\BIBentryALTinterwordstretchfactor}{4}
\providecommand{\BIBentryALTinterwordspacing}{\spaceskip=\fontdimen2\font plus
\BIBentryALTinterwordstretchfactor\fontdimen3\font minus
  \fontdimen4\font\relax}
\providecommand{\BIBforeignlanguage}[2]{{%
\expandafter\ifx\csname l@#1\endcsname\relax
\typeout{** WARNING: IEEEtran.bst: No hyphenation pattern has been}%
\typeout{** loaded for the language `#1'. Using the pattern for}%
\typeout{** the default language instead.}%
\else
\language=\csname l@#1\endcsname
\fi
#2}}
\providecommand{\BIBdecl}{\relax}
\BIBdecl

\bibitem{luby2007rmd}
M.~Luby, T.~Gasiba, T.~Stockhammer, and M.~Watson, ``{Reliable multimedia
  download delivery in cellular broadcast networks},'' \emph{IEEE Trans.
  Broadcasting}, vol.~53, no. 1 Part 2, pp. 235--246, 2007.

\bibitem{berger2008optimizing}
C.~Berger, S.~Zhou, Y.~Wen, P.~Willett, and K.~Pattipati, ``{Optimizing joint
  erasure-and error-correction coding for wireless packet transmissions},''
  \emph{IEEE Transactions on Wireless Communications}, vol.~7, no. 11 Part 2,
  pp. 4586--4595, 2008.

\bibitem{Courtade}
T.~A. Courtade and R.~D. Wesel, ``{A cross-layer perspective on rateless coding
  for wireless channels},'' \emph{Proc. of IEEE Int'l Conf. in Commun.
  (ICC'09)}, pp. 1--6, Jun. 2009.

\bibitem{ChenSub}
X.~Chen, V.~Subramanian, and D.~J. Leith., ``{PHY modulation/rate control for
  fountain codes in 802.11 WLANs},'' \emph{submitted to IEEE Trans. Wireless
  Comm., June 2009}.

\bibitem{CaTaBi}
G.~Carie, G.~Taricco, and E.~Biglieri, ``Optimum power control over fading
  channels,'' \emph{IEEE Trans. Inform. Theory}, vol.~45, no.~5, pp.
  1468--1489, Jul. 1999.

\bibitem{FabregasCaire}
{A. Guill\'en i F\`abregas} and {G. Caire}, ``Coded modulation in the
  block-fading channel: coding theorems and code construction,'' \emph{IEEE
  Trans. Inf. Theory}, vol.~52, no.~1, pp. 91--114, Jan. 2006.

\bibitem{PrasadVaranasi_IT06}
N.~Prasad and M.~K. Varanasi, ``Outage theorems for {MIMO} block-fading
  channels,'' \emph{IEEE Trans. Inform. Theory}, vol.~52, no.~12, pp.
  5284--5296, Dec. 2006.

\bibitem{Malkamaki}
E.~Malkam\"aki and H.~Leib, ``Coded diversity on block-fading channels,''
  \emph{IEEE Trans. Inf. Theory}, vol.~45, no.~2, pp. 771--781, Mar. 1999.

\bibitem{TseVis}
D.~Tse and P.~Viswanath, \emph{Fundamentals of Wireless Communications}.\hskip
  1em plus 0.5em minus 0.4em\relax Cambridge University, 2005.

\bibitem{CaireTuninetti}
G.~Caire and D.~Tuninetti, ``The throughput of {hybrid-ARQ} protocols for the
  {Gaussian} collision channel,'' \emph{IEEE Trans. Inform. Theory}, vol.~47,
  no.~4, pp. 1971--1988, Jul. 2001.

\bibitem{corless1996lambertw}
R.~Corless, G.~Gonnet, D.~Hare, D.~Jeffrey, and D.~Knuth, ``{On the LambertW
  function},'' \emph{Advances in Computational Mathematics}, vol.~5, no.~1, pp.
  329--359, 1996.

\bibitem{Smith_Shafi}
P.~J. Smith and M.~Shafi, ``On a {Gaussian} approximation to the capacity of
  wireless {MIMO} systems,'' \emph{Proc. of IEEE Int'l Conf. in Commun.
  (ICC'02)}, pp. 406--410, Apr. 2002.

\bibitem{Barriac}
G.~Barriac and U.~Madhow, ``Characterizing outage rates for space-time
  communication over wideband channels,'' \emph{IEEE Trans. Commun.}, vol.~52,
  no.~4, pp. 2198--2207, Dec. 2004.

\bibitem{wu2008}
P.~Wu and N.~Jindal, ``{Performance of hybrid-ARQ in block-fading channels: a
  fixed outage probability analysis},'' \emph{to appear at IEEE Trans.
  Commun.}, vol.~58, no.~4, Apr. 2010.

\bibitem{Alouini}
M.~S. Alouini and A.~J. Goldsmith, ``Capacity of {Rayleigh} fading channels
  under different adaptive transmission and diversity-combining techniques,''
  \emph{IEEE Trans. Veh. Technol.}, vol.~48, no.~4, pp. 1165--1181, Jul. 1999.

\bibitem{McKay}
M.~R. McKay, P.~J. Smith, H.~A. Suraweera, and I.~B. Collings, ``On the mutual
  information distribution of {OFDM-based} spatial multiplexing: exact variance
  and outage approximation,'' \emph{IEEE Trans. Inform. Theory}, vol.~54,
  no.~7, pp. 3260--3278, Jul. 2008.

\bibitem{buckingham2008information}
D.~Buckingham and M.~Valenti, ``{The information-outage probability of
  finite-length codes over AWGN channels},'' \emph{42nd Annunal Conf. Inform.
  Sciences and Systems (CISS'08)}, pp. 390--395, 2008.

\bibitem{laneman2006distribution}
J.~Laneman, ``{On the distribution of mutual information},'' \emph{Proc.
  Workshop on Information Theory and its Applications (ITA'06)}, 2006.

\bibitem{Gamal}
H.~E. Gamal, G.~Caire, and M.~E. Damen, ``The {MIMO} {ARQ} channel:
  diversity-multiplexing-delay tradeoff,'' \emph{IEEE Trans. Inform. Theory},
  vol.~52, no.~8, pp. 3601--3621, Aug. 2006.

\bibitem{wolff1989stochastic}
R.~Wolff, \emph{{Stochastic Modeling and the Theory of Queues}}.\hskip 1em plus
  0.5em minus 0.4em\relax Prentice hall, 1989.

\bibitem{A.Goldsmith}
A.~Goldsmith, \emph{Wireless Communications}.\hskip 1em plus 0.5em minus
  0.4em\relax Cambridge University Press, 2005.

\bibitem{meyer2006arq}
M.~Meyer, H.~Wiemann, M.~Renfors, J.~Torsner, and J.~Cheng, ``{ARQ concept for
  the UMTS Long-Term Evolution},'' \emph{IEEE 64th Vehicular Technology
  Conference (VTC'06)}, pp. 1--5, Sep. 2006.

\bibitem{ekstrom2006technical}
H.~Ekstrom, A.~Furuskar, J.~Karlsson, M.~Meyer, S.~Parkvall, J.~Torsner, and
  M.~Wahlqvist, ``{Technical solutions for the \emph{3G} long-term
  evolution},'' \emph{IEEE Communications Magazine}, vol.~44, no.~3, pp.
  38--45, 2006.

\bibitem{NiharAngel}
A.~Lozano and N.~Jindal, ``Transmit diversity v. spatial multiplexing in modern
  {MIMO} systems,'' \emph{IEEE Trans. Wireless Commun.}, vol.~9, no.~1, pp.
  186--197, Jan. 2010.

\bibitem{apostol1974mathematical}
T.~Apostol, \emph{{Mathematical Analysis}}.\hskip 1em plus 0.5em minus
  0.4em\relax Addison-Wesley Reading, MA, 1974.

\bibitem{Nick}
N.~Kingsbury, ``Approximation formula for the {Gaussian} error integral,
  {Q(x)},'' http://cnx.org/content/m11067/latest/.

\end{thebibliography}

\newpage

\begin{figure}[ht]
\begin{center}
\includegraphics[width = 3.9in]{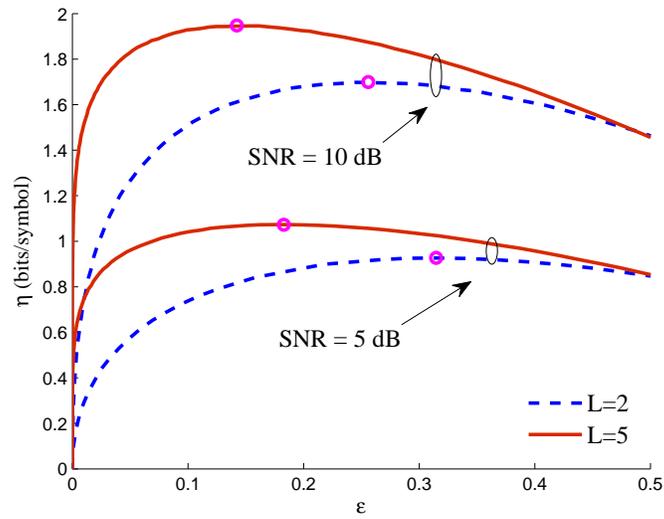}
\caption{Gooput $\eta$ (bits/symbol) vs. PHY outage probability $\varepsilon$ for $L=2,5$, $\snr=10$ dB} \label{fig:gput_eps}
\end{center}
\end{figure}

\begin{figure}[ht]
\begin{center}
\includegraphics[width = 3.9in]{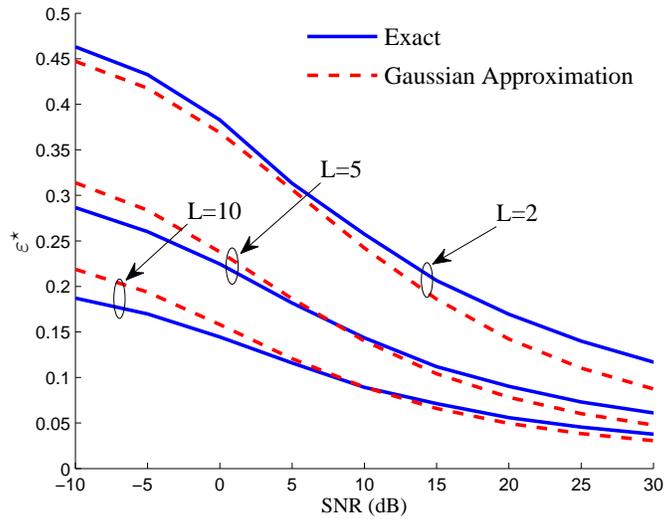}
\caption{Optimal $\varepsilon$ vs. $\snr$ (dB) for $L=2,5,10$} \label{fig:opt_eps_SNR}
\end{center}
\end{figure}

\begin{figure}[ht]
\centering
\subfigure[$1-\varepsilon$ vs. $R_{\varepsilon}$ (bits/symbol)]{\label{fig:suc_pr_R}
\includegraphics[width = 3.1in]{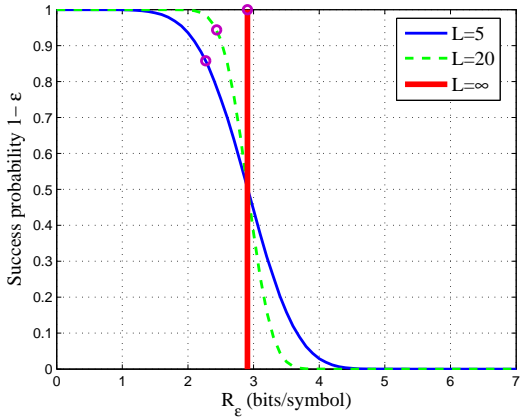}}
\subfigure[$\eta$ (bits/symbol) vs. $R_{\varepsilon}$ (bits/symbol)]{\label{fig:goodput_eps}
\includegraphics[width = 3.1in]{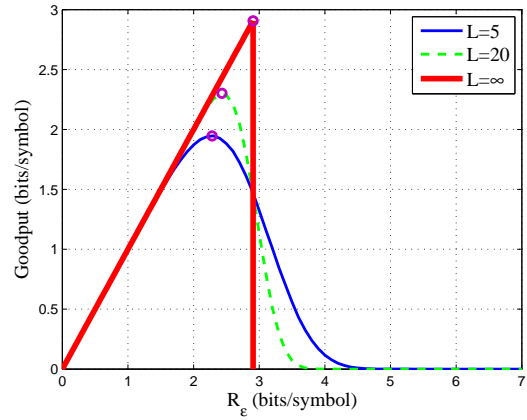}}
\caption{Success probability $1-\varepsilon$ and $\eta$ (bits/symbol) vs. $R_{\varepsilon}$ (bits/symbol) for $\snr=10$ dB}\label{fig:phy_arg}
\end{figure}

\begin{figure}[ht]
\centering
\subfigure[$L=2$]{\label{fig:goodput_snr_2}
\includegraphics[width = 3.1in]{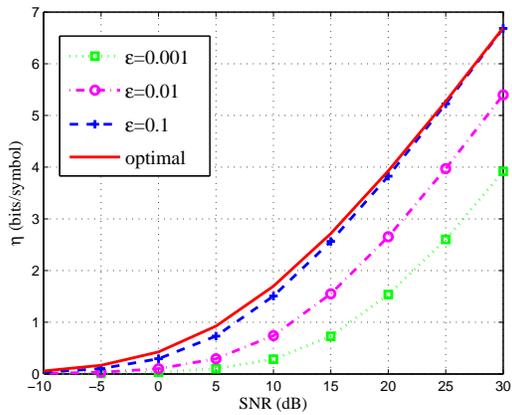}}
\subfigure[$L=10$]{\label{fig:goodput_snr_10}
\includegraphics[width = 3.1in]{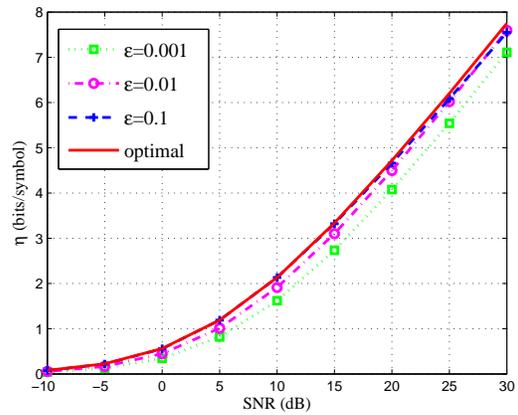}}
\caption{$\eta$ (bits/symbol) vs. $\snr$ (dB),
for $\varepsilon=0.001,0.01,0.1$, and
$\varepsilon^{\star}$}\label{fig:goodput_snr}
\end{figure}

\begin{figure}[ht]
\begin{center}
\includegraphics[width = 4.5in]{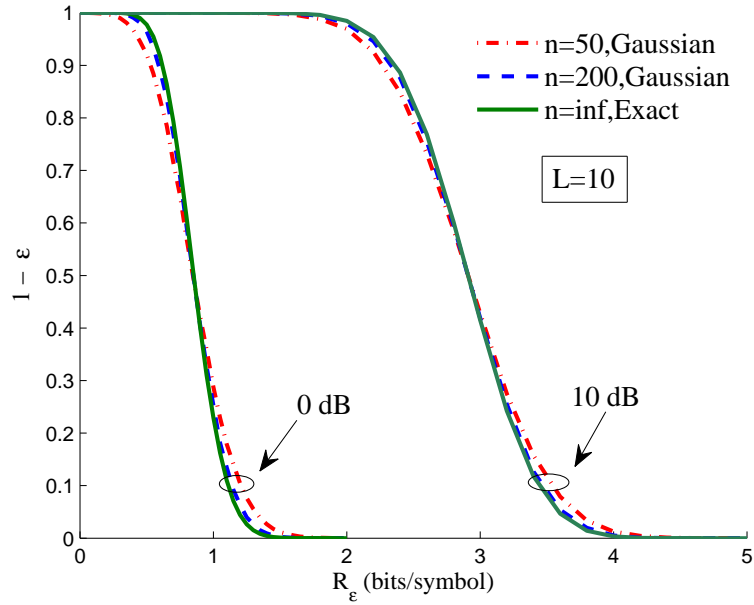}
\caption{Success probability $1-\varepsilon$  vs. transmitted rate $R_{\varepsilon}$ (bits/symbol) for $n=50,200,\infty$, $L=10$ at $\snr=0$ and $10$ dB} \label{fig:ps_R_fin}
\end{center}
\end{figure}

\begin{figure}[ht]
\begin{center}
\includegraphics[width = 4.5in]{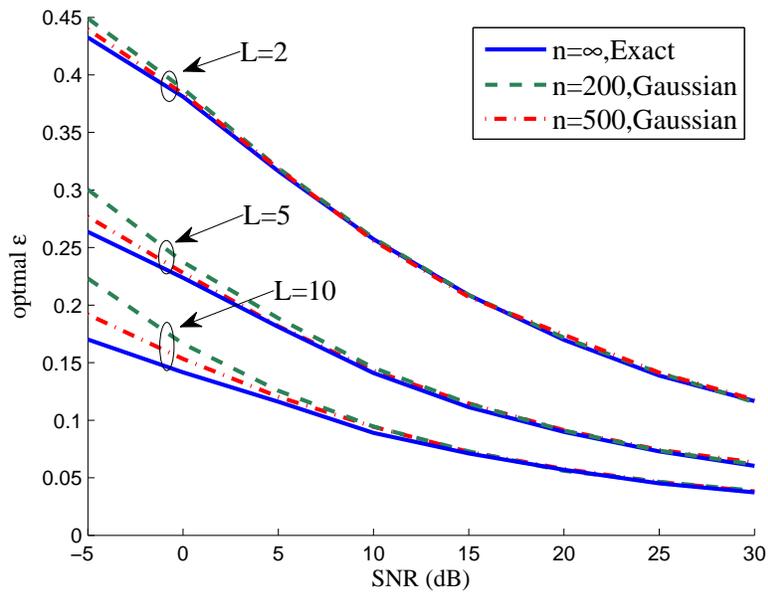}
\caption{Optimal $\varepsilon$  vs. $\snr$ (dB) for $L=2,5,10$ and $n=200, 500$ and $\infty$} \label{fig:opt_eps_snr_fin}
\end{center}
\end{figure}

\begin{figure}[ht]
\begin{center}
\includegraphics[width = 5in]{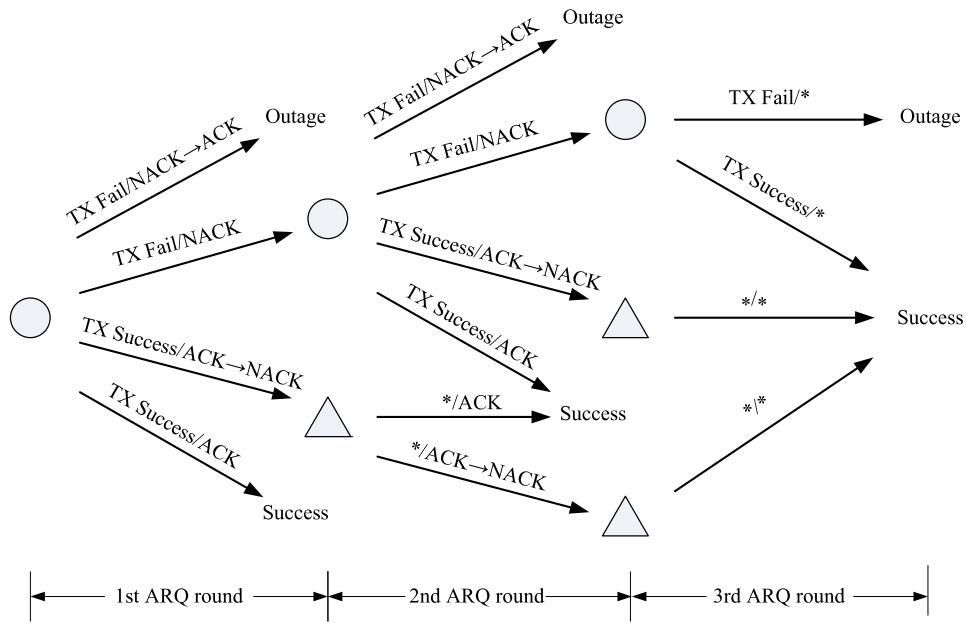}
\caption{The ARQ process with non-ideal feedback with an end-to-end delay constraint $d=3$.} \label{fig:arq_proc}
\end{center}
\end{figure}

\begin{figure}[ht]
\begin{center}
\includegraphics[width = 4.4in]{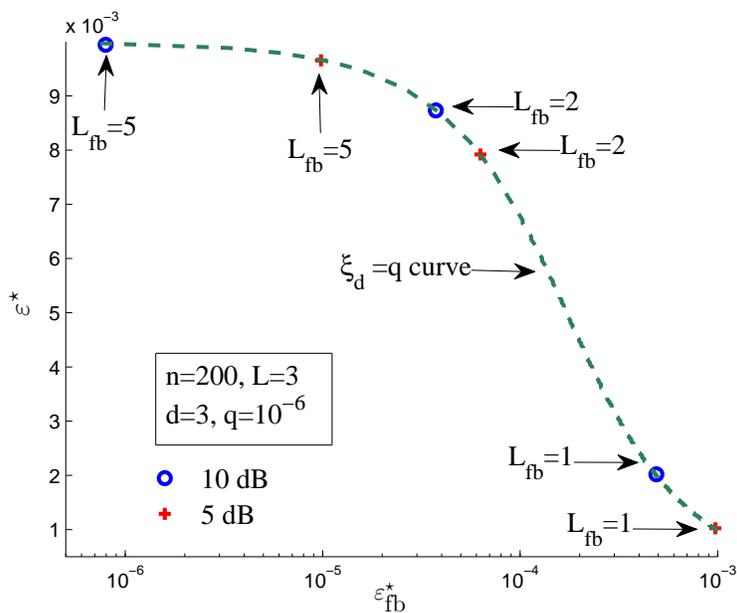}
\caption{($\varepsilon^{\star},\varepsilon_{\textrm{fb}}^{\star}$) with $L_{\textrm{fb}}=1,2$ and $5$ in Rayleigh fading feedback channel for $n=200$, $d=3$, $q=10^{-6}$, and $L=3$ at $\snr=5$ and $10$ dB. The curve specifying the ($\varepsilon,\varepsilon_{\textrm{fb}}$) pairs that achieve the reliability constraint $\xi_d = q$ is also plotted.} \label{fig:eps_theta_del}
\end{center}
\end{figure}

\begin{figure}[ht]
\begin{center}
\includegraphics[width = 4.4in]{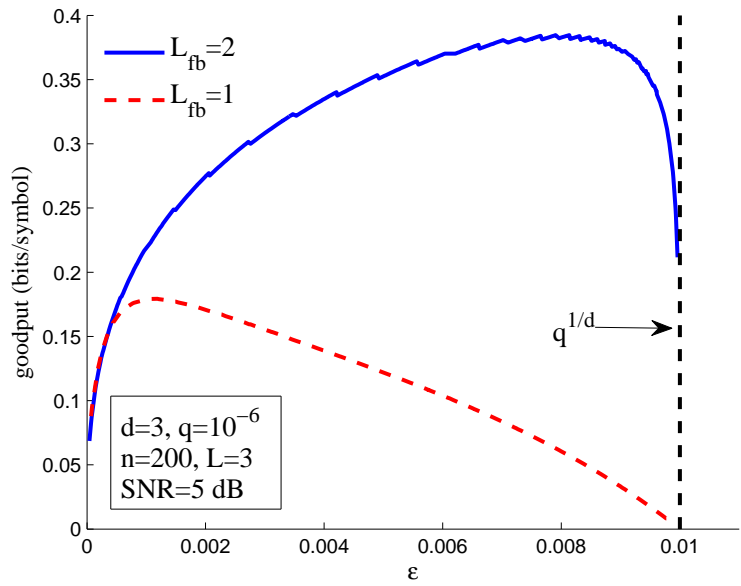}
\caption{Goodput $\eta$ (bits/symbol) vs. PHY outage probability $\varepsilon$ with $L_{\textrm{fb}}=1$ and $2$ in Rayleigh fading feedback channel for $\snr=5$ dB, $n=200$, $L=3$, $d=3$ and $q=10^{-6}$. } \label{fig:gput_eps_del}
\end{center}
\end{figure}


\begin{figure}[ht]
\begin{center}
\includegraphics[width = 4.5in]{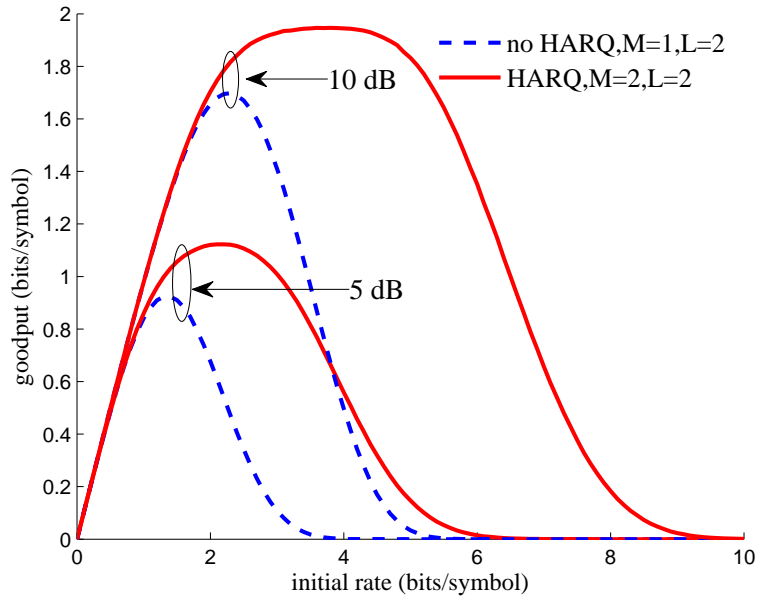}
\caption{Goodput (bits/symbol)  vs. initial rate (bits/symbol) with HARQ for $M=2$ and $L=2$ and without HARQ for $M=1$ and $L=2$ at $\snr=5,10$ dB. } \label{fig:gput_r_comp}
\end{center}
\end{figure}

\end{document}